\documentclass[preprint,12pt]{elsarticle}

\usepackage[T1]{fontenc}
\usepackage[utf8]{inputenc}
\usepackage{lmodern}
\usepackage{amsmath,amssymb,amsthm,mathtools}
\usepackage{graphicx}
\usepackage{booktabs}
\usepackage{array}
\usepackage{multirow}
\usepackage{siunitx}
\usepackage{xcolor}
\usepackage{microtype}
\usepackage{enumitem}
\usepackage{algorithm}
\usepackage{algpseudocode}
\usepackage{caption}
\usepackage{subcaption}
\usepackage{float}
\usepackage[hidelinks]{hyperref}
\usepackage[nameinlink,noabbrev]{cleveref}

\graphicspath{{./}}
\setlist{nosep,leftmargin=*}

\hypersetup{
  pdftitle={Witnessed Boolean Contracts: Sound Local Distance Guarantees for Approximate Implicit CSG},
  pdfauthor={Jeremy Pei},
  pdfsubject={Certified geometric contracts for approximate signed-distance Boolean operations},
  pdfkeywords={signed distance field, constructive solid geometry, implicit modelling, geometric certification, interval arithmetic, witness point}
}

\theoremstyle{plain}
\newtheorem{theorem}{Theorem}
\newtheorem{lemma}{Lemma}

\newtheorem{corollary}{Corollary}
\theoremstyle{definition}
\newtheorem{definition}{Definition}

\theoremstyle{remark}
\newtheorem{remark}{Remark}

\newcommand{\R}{\mathbb{R}}
\newcommand{\sd}{\operatorname{sd}}
\newcommand{\dist}{\operatorname{dist}}
\newcommand{\reach}{\operatorname{reach}}
\newcommand{\argmin}{\operatorname*{arg\,min}}
\newcommand{\eps}{\varepsilon}
\newcommand{\ph}{\widehat{\phi}}
\newcommand{\nh}{\widehat{n}}
\newcommand{\pp}{p^{\star}}

\newcommand{\Cert}{\mathsf{Cert}}
\newcommand{\Unknown}{\mathsf{Unknown}}

\newcolumntype{P}[1]{>{\raggedright\arraybackslash}p{#1}}

\begin{document}

\begin{frontmatter}

\title{Witnessed Boolean Contracts: Sound Local Distance Guarantees for
Approximate Implicit CSG}

\author{Jeremy Pei\corref{cor1}}
\ead{jypei@nuaa.edu.cn}
\affiliation{organization={Independent Researcher}}
\cortext[cor1]{Corresponding author.}

\begin{abstract}
Hard min/max operations preserve the sign of exact implicit solids but generally
destroy signed-distance semantics. A value-error budget alone cannot repair the
problem: arbitrarily small, sign-preserving perturbations can rotate the
normalized gradient by nearly a right angle and move a gradient-projected
boundary witness by a finite amount. We formalize this obstruction as the
absence of a value-only stability modulus for witness projection.

We then introduce a \emph{witnessed Boolean contract}. At each query, the
contract checks five items: the value error, the gradient error, positive
reach, separation of the active branch, and a Lipschitz visibility margin at
the approximate witness. A sharp normalization lemma turns the gradient-error
bound into a bound on the normal direction. We enlarge the local neighborhood
so that the inactive operand can be checked at the true foot point as well as
at the approximate witness. Under the stated analytic and domain assumptions,
passing these checks proves that the selected branch gives a true point on the
union boundary. It also bounds the distance deficiency by the active
operand's pointwise value error. If any check fails or cannot be evaluated, the
result is \(\Unknown\).
The same argument extends to any finite hard union by checking the selected
witness against every inactive operand.

An executable audit checks the analytic budgets, derivatives, wrapped interval
bounds, branch order, and each step of the proof on \num{75738} certified
samples; it finds no violations. A broader stress test uses spheres, capsules,
ring tori, three Boolean operators, and a four-leaf union. Among \num{96000}
analytic narrow-band samples, \num{87469} are certified with zero oracle
violations; per-case coverage ranges from \SI{72.67}{\percent} to
\SI{98.30}{\percent}. When the visibility check is removed, 75 truly invalid
witnesses are accepted. Thus branch separation alone is not enough. The
non-union rows are labeled as empirical tests, not as new theorems.
\end{abstract}


\end{frontmatter}

\section{Introduction}
\label{sec:introduction}

Implicit CAD describes a solid \(X\subset\R^3\) with a scalar field: the sign
tells whether a point is inside or outside.  If the field is the signed
distance \(\phi_X=\sd(\cdot,\partial X)\), its magnitude gives the distance to
the boundary, and its gradient gives a surface normal almost everywhere.
These properties are used in sphere tracing, offsets, narrow-band sampling,
collision queries, and Hermite-data extraction for dual
contouring~\cite{hart1996sphere,frisken2000adaptive,ju2002dual,museth2013vdb}.

The standard hard union
\begin{equation}
  \psi(x)=\min\{\phi_A(x),\phi_B(x)\}
  \label{eq:hard-union}
\end{equation}
gives the correct union sign for exact operands, but \(\psi\) is not generally
the signed distance to \(A\cup B\).  In an overlap, the closest point on the
selected primitive can lie inside the other solid and therefore be hidden from
the union boundary.  This is a familiar source of pseudo-distance behavior in implicit modelling
\cite{bloomenthal1997introduction,pasko1995function,wyvill1999extending}, but
the usual remedies---adding a global error label or re-distancing the field---do
not answer a local question: \emph{is the projected point on the selected
primitive really a point on the Boolean boundary?}

This question shows that value accuracy and geometric stability are different.
If an approximate field \(\ph\) satisfies
\(\|\ph-\phi\|_\infty\le\eps\), its normalized gradient can still be nearly
orthogonal to the true normal.  A witness computed as
\(x-\ph\nabla\ph/\|\nabla\ph\|\) can therefore remain far from the correct point even as
\(\eps\to0\).  Consequently, neither a value budget nor sign preservation can
justify a witness test by itself.

This paper gives a conditional guarantee for this local question.  Instead of
assigning one finite error bound to the whole Boolean field, it checks each
query point separately.  The test uses quantities supplied by an upstream
certified implementation: value-error budgets for the operands, a gradient
error bound for the cell, branch separation, a reach/tube check, and a
Lipschitz bound on an enlarged box.  If any required quantity is unknown, the
test returns \(\Unknown\).

\paragraph{Contributions.}
\begin{enumerate}
  \item We prove that value accuracy and sign preservation admit no stability
        modulus for normalized-gradient witness projection
        (\Cref{thm:value-insufficient}).  The counterexample remains inside a
        positive-reach tube and keeps the field value arbitrarily accurate.
  \item We derive a tight bound for normalizing a perturbed unit vector
        (\Cref{lem:normal}) and combine it with value error to bound the
        approximate witness displacement (\Cref{lem:projection}).
  \item We formulate a witnessed Boolean certificate for consuming upstream
        value, gradient, and range certificates, and prove its soundness
        for binary hard union (\Cref{thm:witness-sound}) and any finite hard
        union (\Cref{cor:nary-union}).  In contrast to a global
        pseudo-SDF label, the theorem identifies the precise local obligations
        that make a finite deficiency bound valid.
  \item We provide a reproducible proof-chain audit.  It checks derivatives,
        wrapped interval arithmetic, branch order, and four analytic links on
        \num{75738} certified samples, while explicitly separating proof,
        implementation validation, coverage, and empirical completeness.
        A second benchmark tests three analytic primitive classes, three
        Boolean operators, a four-leaf union, running time, reasons for
        rejection, and guard ablations.  The four-leaf case tests the n-ary
        corollary.  Intersection and difference are reported only as empirical
        tests.
\end{enumerate}

The result is deliberately narrower than a new distance reconstruction
method.  It does not make the hard minimum a global SDF.  It makes a local
geometric statement checkable and safe to reject.  Composition over mixed
Boolean trees is not claimed here.

\section{Related work}
\label{sec:related}

\paragraph{Distance fields and implicit CSG.}
Jones et al.~\cite{jones2006distance} survey how distance fields are built and
used.  Early function-based CSG expressed union and intersection by combining
scalar functions~\cite{ricci1973constructive}.  R-functions later provided a
systematic algebra that links logic to geometry~\cite{shapiro2007semianalytic}.
Classical implicit modelling combines real functions with Boolean, blending,
and warping operators
\cite{bloomenthal1997introduction,pasko1995function,wyvill1999extending}.
These operators usually preserve the represented solid more easily than they
preserve distance values.  B\'alint et al.~\cite{balint2023operations} analyze operations on
signed-distance estimates and are especially close to our concern with
one-sided guarantees.  Our contribution is complementary: we certify a
nearest-point witness on a local subdomain and include normal-direction error
and branch stability in the contract.

Barbier et al.'s Lipschitz Pruning~\cite{barbier2025pruning} is a particularly
close recent neighbor.  It uses regional Lipschitz reasoning to replace a CSG
subtree by an equivalent active operand, thereby accelerating evaluation of
large primitive trees.  Our regional question and output guarantee are
different: the expression may already have selected the correct operand at
the query while its projected foot point is hidden by another solid.  We bound
the displacement of that foot point under approximate values and normals, then
certify Boolean-boundary visibility and distance deficiency.  The comparison
is summarized in \Cref{tab:related-comparison}.

\begin{table}[H]
\centering
\caption{Relation to selected closest guarantees.  ``Regional'' means that a
claim can vary over space; it does not imply that the certified object is the
same.}
\label{tab:related-comparison}
\small
\begin{tabular}{@{}P{0.17\linewidth}P{0.24\linewidth}P{0.24\linewidth}P{0.26\linewidth}@{}}
\toprule
 & B\'alint et al.~\cite{balint2023operations}
 & Lipschitz Pruning~\cite{barbier2025pruning}
 & This work\\
\midrule
Primary goal
& Precision and sphere-tracing convergence of SDF estimates under set operations
& Faster regional evaluation by pruning equivalent CSG subtrees
& Sound visibility of a projected primitive foot point\\
Operand model
& Signed-distance estimates and lower bounds
& Analytic/conservative 1-Lipschitz primitive trees
& Approximate operands with explicit value and gradient budgets\\
Regional object
& Accuracy described through spatial cases and set-contact smoothness
& Active nodes or an equivalent local expression
& Query witness, inflated transfer box, and branch-separation domain\\
Differential uncertainty
& Geometry-dependent precision analysis
& Lipschitz bounds for expression pruning
& Explicit normalized-gradient direction error and witness displacement\\
Output guarantee
& Precision/convergence bounds for the represented field
& Same local expression value, or a conservative far-field constant
& True Boolean-boundary witness and a deficiency budget, otherwise \(\Unknown\)\\
\bottomrule
\end{tabular}
\end{table}

The same research line constructs exact polygonal signed-distance fields from
cut-locus and Voronoi decompositions~\cite{balint2024cutlocus}.  That work
locates the structures governing multiple foot points and supplies exact
operands.  We address a different case: the operands may be approximate, and
we certify one Boolean witness at a time.  For exact operands, the value- and
gradient-error terms disappear, but the reach and visibility checks are still
needed.

\paragraph{Why prior guarantees do not imply witness soundness.}
Field-precision bounds such as B\'alint et al.'s control field values and
tracing behavior, but they do not locate the foot point obtained from the
normalized gradient.  Lipschitz Pruning proves that two expressions agree in
a region; it does not prove that the selected primitive foot point remains
visible after the Boolean operation.  Closest-point regularizers
and recent consistent SDF interpolation methods repair or reconstruct a field
rather than certify a query on the unmodified approximate Boolean result
\cite{marschner2023constructive,chen2026greed}.  Likewise, gradient-based SDF
contouring constructs and filters candidate contact points but does not
propagate operand value and gradient uncertainty through a Boolean visibility
test~\cite{kohlbrenner2026contouring}.  None of these conclusions supplies the
implication proved here: acceptance entails a true Boolean-boundary witness
and a pointwise deficiency budget.

\paragraph{Certified evaluation of implicit fields.}
Sphere tracing and guaranteed ray intersection rely on conservative field or
Lipschitz information rather than on a representation label alone
\cite{kalra1989guaranteed,hart1996sphere}.  Segment Tracing and Generalized
Lipschitz Tracing use local or directional Lipschitz upper bounds to choose
safe ray steps~\cite{galin2020segment,ban2025generalized}; Lipschitz Pruning
uses such bounds to establish regional expression equivalence.  In contrast,
our transfer lemma uses \(L_b\) to move an inactive-operand value guarantee
across the inflated box enclosing witness uncertainty.  Interval analysis
provides a
standard way to enclose function ranges and derivatives over cells
\cite{moore1966interval,snyder1992interval}; Duff's interval subdivision
applies the same validated reasoning directly to implicit CSG
\cite{duff1992interval}.  Range analysis has also enabled
guaranteed queries on neural implicit surfaces~\cite{sharp2022spelunking}.
Our cell-wise phase bounds follow this tradition, but the certified predicate
is different: it transfers visibility of a projected primitive foot point
through an approximate Boolean computation.

\paragraph{SDF regularity and learned fields.}
Marschner et al.~\cite{marschner2023constructive} characterize pseudo-SDFs
created by CSG operations on neural signed-distance fields and regularize
closest-point behavior to improve the resulting field.  Their method is a
potential upstream producer of better operands; our contract instead consumes
explicit local budgets and certifies a downstream Boolean-boundary witness.

Recent work has emphasized that an Eikonal residual does not by itself certify
distance correctness and has developed improved regularizers or Lipschitz
architectures~\cite{yang2023steik,wang2025hotspot,coiffier2024lipschitz}.
Those methods aim to construct a better field.  We instead assume that a field
comes with explicit local budgets and ask what downstream Boolean statement
can be proved from them.

\paragraph{Reach and nearest-point regularity.}
The reach of a set controls uniqueness and regularity of nearest-point
projection in a tubular neighborhood~\cite{federer1959curvature}.  Positive
reach is central to our theorem: it makes the true active-operand foot point
well defined.  Sharp features have zero reach and are outside the theorem's
scope; this is reported as a limitation rather than hidden by sampling.

\section{Preliminaries and problem statement}
\label{sec:preliminaries}

\subsection{Signed distance and Boolean deficiency}

We take \(\phi_X<0\) inside a solid \(X\), \(\phi_X=0\) on
\(\partial X\), and \(\phi_X>0\) outside.  For two solids \(A,B\), their
approximate hard-union field is
\begin{equation}
  \psi(x)=\min\{\ph_A(x),\ph_B(x)\}.
\end{equation}
Let \(a(x)\in\{A,B\}\) be the branch selected by the approximate fields, and
let \(b(x)\) be the other branch.  When the selected branch is unambiguous, we
define the local distance deficiency as
\begin{equation}
  \Delta^-(x)
  =\dist\!\left(x,\partial(A\cup B)\right)-|\psi(x)|.
  \label{eq:deficiency}
\end{equation}
A positive \(\Delta^-\) means that the field under-reports the distance to the
Boolean boundary.

\begin{definition}[Value/sign contract]
An approximate operand \(\ph_i\), \(i\in\{A,B\}\), satisfies a value/sign
contract on \(D_i^{\mathrm{val}}\) if
\begin{equation}
  |\ph_i(y)-\phi_i(y)|\le\eps_i(y),\qquad
  \operatorname{sign}\ph_i(y)=\operatorname{sign}\phi_i(y)
  \quad (y\in D_i^{\mathrm{val}}).
  \label{eq:value-contract}
\end{equation}
The error budget can be a single global value,
\(\eps_i(y)\equiv\eps_i\), or a value certified separately at each point.
\end{definition}

\subsection{Witnesses}

Suppose branch \(a\) is selected at \(x\).  At a point where the true SDF is
differentiable, write \(n_a(x)=\nabla\phi_a(x)\), with \(\|n_a\|=1\).  The
true foot point on the primitive and the approximate gradient witness are
\begin{align}
  \pp(x)&=x-\phi_a(x)n_a(x), \\
  \widehat p(x)&=x-\ph_a(x)\nh_a(x),
  &\nh_a(x)&=\frac{\nabla\ph_a(x)}{\|\nabla\ph_a(x)\|}.
  \label{eq:witnesses}
\end{align}
For exact operands, \(\pp\) is a valid point on the union boundary exactly
when the other solid does not hide it.  With approximate operands, both the
location of \(\widehat p\) and the inactive field value there are uncertain.

\subsection{Three correctness notions}

We use three different terms throughout the paper:
\begin{itemize}
  \item \emph{soundness}: certificate acceptance implies the geometric claim;
  \item \emph{completeness}: every point satisfying the geometric claim is
        accepted;
  \item \emph{coverage}: the fraction of sampled target points accepted.
\end{itemize}
Finding no violations is useful for checking the implementation, but it does
not prove soundness; the analytic implication provides that proof.  Similarly,
\SI{100}{\percent} coverage on a finite sample does not prove completeness on
the full continuous domain.

\section{Value accuracy does not stabilize witnesses}
\label{sec:impossibility}

The first result explains why a witnessed contract must contain differential
information.

\begin{theorem}[No value-only witness stability modulus]
\label{thm:value-insufficient}
Fix \(0<h_0<\rho\).  There is no function
\(m(\eps,\rho)\) with \(m(\eps,\rho)\to0\) as \(\eps\to0\) such that every
sign-preserving approximation satisfying
\(\|\ph-\phi\|_{\infty}\le\eps\) obeys
\begin{equation}
  \|\widehat p(x)-\pp(x)\|\le m(\eps,\rho)
  \label{eq:modulus}
\end{equation}
at all points with \(|\phi(x)|=h_0\) in the reach tube.
\end{theorem}

\begin{proof}
Take the planar SDF \(\phi(x)=x_1\) on a bounded box, where
\(\max|\phi|=M\).  Its zero set has infinite reach.  Define
\begin{equation}
  \ph_k(x)=\phi(x)\bigl(1+\alpha_k\sin(\omega_k x_2)\bigr),
  \qquad \alpha_k=\eps_k/M<1,
\end{equation}
where \(\eps_k\downarrow0\) and \(\alpha_k\omega_k h_0\to\infty\).
The multiplier is positive.  Thus the sign is unchanged and
\(\|\ph_k-\phi\|_\infty\le\eps_k\).  At
\(x_0=(h_0,0,0)\), \(\ph_k(x_0)=\phi(x_0)=h_0\), while
\begin{equation}
  \nabla\ph_k(x_0)=e_1+\alpha_k\omega_k h_0e_2.
\end{equation}
Thus \(\nh_k(x_0)\to e_2\), orthogonal to the true normal \(e_1\).  From
\eqref{eq:witnesses},
\begin{equation}
  \|\widehat p_k(x_0)-\pp(x_0)\|
  =h_0\|\nh_k(x_0)-e_1\|\longrightarrow\sqrt{2}\,h_0,
\end{equation}
although \(\eps_k\to0\).  Hence no modulus satisfying
\eqref{eq:modulus} exists.
\end{proof}

\begin{corollary}[Explicit Boolean witness instability]
\label{cor:boolean-instability}
One can choose positive-reach solids \(A\) and \(B\) with \(A\cup B\neq A\) and
sign-preserving approximations with vanishing value error so that the true and
approximate branch choices agree.  Nevertheless, the true foot point lies on
\(\partial(A\cup B)\) while the approximate witness lies inside \(B\).
Therefore a value/sign contract does not by itself validate a Boolean witness.
\end{corollary}

\begin{proof}
Let the active solid be the half-space \(A=\{x_1\le0\}\), use the query
\(x_0=(-h_0,0,0)\), and apply the perturbation from
\Cref{thm:value-insufficient} in the \(x_2\) direction.  The true foot point is
\(\pp=(0,0,0)\).  Here
\(\nabla\ph_k(x_0)=e_1-\alpha_k\omega_kh_0e_2\).  Therefore the approximate
witness \(\widehat p_k\) approaches \(q=(-h_0,-h_0,0)\) as
\(\beta_k:=\alpha_k\omega_kh_0\to\infty\); explicitly,
\begin{equation}
  \|\widehat p_k(x_0)-q\|
  =h_0\sqrt{2-2\beta_k/\sqrt{1+\beta_k^2}}\longrightarrow0.
  \label{eq:witness-limit}
\end{equation}

Now choose \(B\) as the closed ball with radius \(r=\tfrac52h_0\) and center
\(c=(-h_0,-3h_0,0)\).  We use its exact SDF, \(\ph_B=\phi_B\).  Its boundary
has positive reach, \(\reach(\partial B)=r\).  The following four facts hold at
the same time.

\emph{(i) The union is non-trivial.}
The rightmost point of \(B\) is
\(c_1+r=\tfrac32h_0>0\), so \(B\not\subset A\) and
\(A\cup B\supsetneq A\).

\emph{(ii) The true foot point is outside \(B\).}
The distance is \(\|\pp-c\|=\sqrt{10}h_0>r\), with slack
\((\sqrt{10}-\tfrac52)h_0\).  Hence
\(\pp\in\partial A\setminus B\subset\partial(A\cup B)\).  Moreover, because
\(A\) contains the entire half-space \(\{x_1\le0\}\), every
\(y\in\partial(A\cup B)\) satisfies \(y_1\ge0\).  Thus
\(\|x_0-y\|\ge h_0\), while equality is attained at \(y=\pp\), and therefore
\begin{equation}
  \dist(x_0,\partial(A\cup B))=h_0=|\psi(x_0)|.
\end{equation}
The query has zero deficiency, so the failure concerns witness validity rather
than the Boolean field value.

\emph{(iii) The approximate witness is inside \(B\).}
Because \(\|q-c\|=2h_0<r\), the limiting witness lies inside the ball with
slack
\(\tfrac12h_0\).  More explicitly,
\begin{equation}
  \frac{\|\widehat p_k-c\|^2}{h_0^2}
  =10-\frac{6\beta_k}{\sqrt{1+\beta_k^2}}<\frac{25}{4}
\end{equation}
whenever \(\beta_k>5/\sqrt{39}\approx0.80064\).  Hence
\(\widehat p_k\in\operatorname{int}B\) for all sufficiently large \(k\).

\emph{(iv) Both branch choices select \(A\).}
Because \(\|x_0-c\|=3h_0\),
\begin{equation}
  \phi_B(x_0)=\ph_B(x_0)=\tfrac12h_0
  >-h_0=\phi_A(x_0)=\ph_A(x_0).
\end{equation}
The query also lies outside \(B\), with slack \(\tfrac12h_0\).

Meanwhile the active value error tends to zero and the inactive field is exact.
Thus all parts of the counterexample hold at the same time.
\end{proof}

\begin{figure}[t]
\centering
\includegraphics[width=\linewidth]{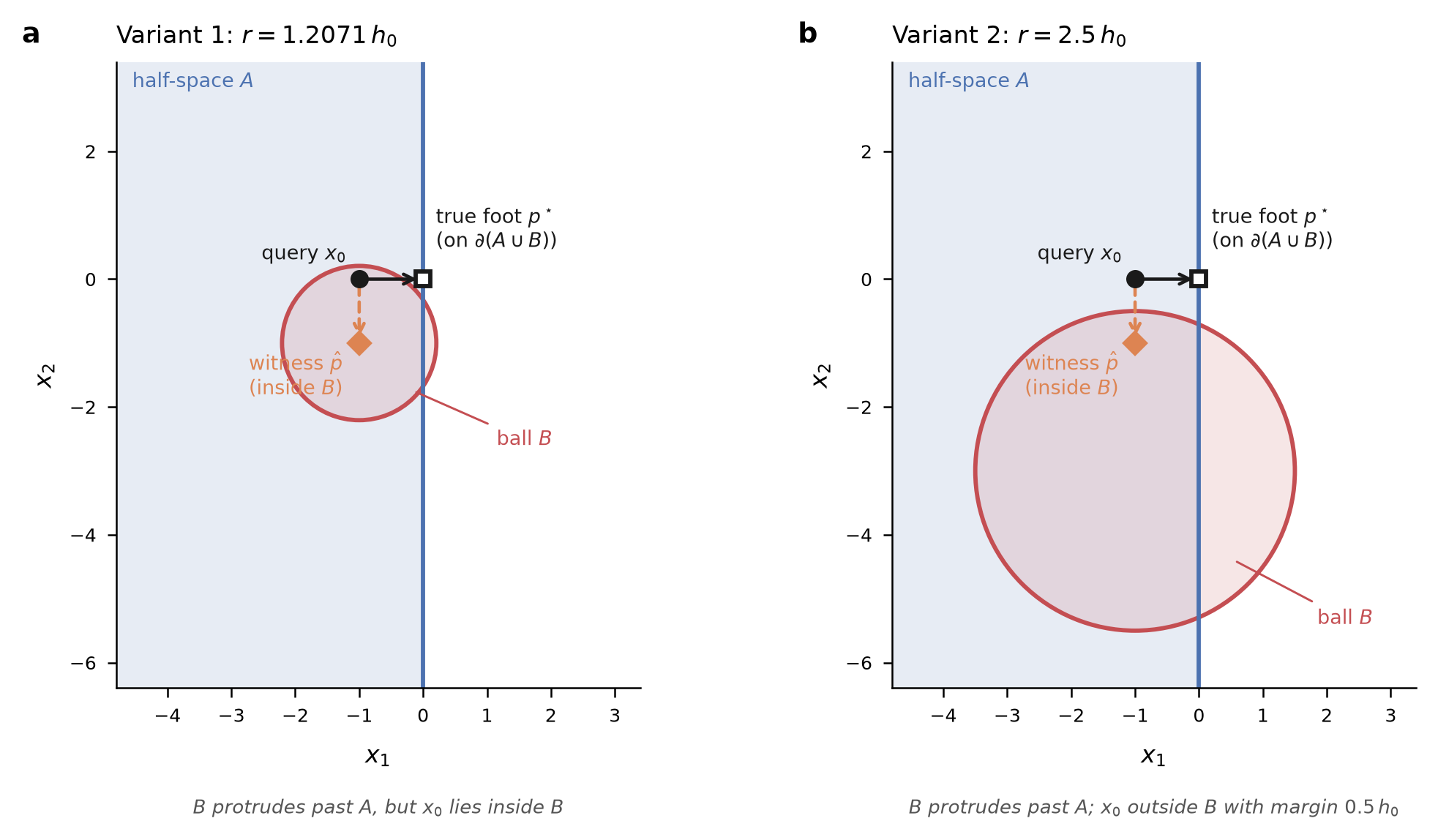}
\caption{Geometry of Corollary~\ref{cor:boolean-instability}, scaled by
\(h_0\).  In both variants the true foot point \(p^\star\) is visible on
\(\partial(A\cup B)\), whereas the perturbed-gradient witness \(\widehat p\)
lies inside \(B\).  The right construction is used in the proof because the
query also remains outside \(B\), with the maximal balanced slack
\(h_0/2\).}
\label{fig:boolean-counterexample}
\end{figure}

\begin{remark}
Keeping the centre at \(q\), the four claims required by the
corollary---non-triviality, visibility of the true foot point, occlusion of the
approximate witness, and branch agreement---hold exactly for
\(r\in(h_0,\sqrt2h_0)\).  The balanced choice is
\(r=\tfrac12(1+\sqrt2)h_0\).  That simpler variant, however, places \(x_0\)
inside \(B\).  Requiring in addition that \(x_0\notin B\) forces the centre
away from \(q\).  If
\(m=\min\{r-\|q-c\|,\,\|x_0-c\|-r\}\) is the smaller of the witness-inside and
query-outside slacks, the triangle inequality gives
\(m\le\tfrac12\|x_0-q\|=\tfrac12h_0\), and the construction above attains this
bound.
\end{remark}

\begin{remark}[Why the certificate rejects the counterexample]
At \(x_0\),
\begin{equation}
  \|\nabla\ph_k(x_0)-\nabla\phi(x_0)\|=\beta_k,
\end{equation}
so A6 already rejects the family once \(\beta_k\ge1\).  Independently, for
all sufficiently large \(k\), \(\widehat p_k\in\operatorname{int}B\), hence
\(\ph_B(\widehat p_k)<0\).  Since \(\eps_B=0\), subtracting any nonnegative
transfer allowance cannot make the visibility expression in
\Cref{sec:certificate} positive.  Thus A6 prevents an uncontrolled witness
displacement, while the visibility guard independently prevents acceptance of
the resulting occluded witness.
\end{remark}

The theorem does not claim that every approximate SDF is unstable.  It only
shows that a value-error budget cannot prove stability by itself.  A sound
certificate must also include normal information, or it must replace the
single witness by a set-valued enclosure.

\section{A witnessed Boolean certificate}
\label{sec:certificate}

We now state sufficient conditions that can be checked at one point.  Assume
\(\ph_A(x)\le\ph_B(x)\), so \(a=A\) and \(b=B\); the other case is symmetric.

\subsection{Assumptions and certificate quantities}

The theorem requires the following conditions.

\begin{description}[style=nextline]
  \item[A1 (positive reach).]
  The active primitive has \(\reach(\partial A)\ge\rho>0\), and its true SDF
  is evaluated in the corresponding tube.

  \item[A2 (value budgets and domain closure).]
  Operand \(i\) has a valid value contract on
  \(D_i^{\mathrm{val}}\), with
  \(|\ph_i-\phi_i|\le\eps_i\).  At the query, the active operand additionally
  has a certified pointwise budget
  \(|\ph_a(x)-\phi_a(x)|\le\delta_a(x)\le\eps_a\).  The query satisfies
  \(x\in D_A^{\mathrm{val}}\cap D_B^{\mathrm{val}}\).  The source cell
  \(C_x\) lies in the validity domain \(D_a^{\nabla}\) of the active gradient
  certificate, and the inflated box introduced below satisfies
  \begin{equation}
    \mathcal{B}_x\subset
    D_b^{\mathrm{val}}\cap D_b^{\mathrm{Lip}}.
    \label{eq:domain-closure}
  \end{equation}

  \item[A3 (sign semantics).]
  Both approximate operands preserve the signs of their referenced solids.
  This is part of the full field contract and supports inside/outside consumers.
  It is \emph{not} a logical premise of the deficiency implication in
  \Cref{thm:witness-sound}, whose proof uses A2 directly.

  \item[A4 (tube stability).]
  With \(h=|\ph_a(x)|\),
  \begin{equation}
    h<\rho,\qquad h+\delta_a(x)<\rho.
    \label{eq:A4}
  \end{equation}

  \item[A5 (certified branch separation).]
  For a fixed \(\gamma>0\),
  \begin{equation}
    \ph_b(x)-\ph_a(x)>\eps_a+\eps_b+\gamma.
    \label{eq:A5}
  \end{equation}

  \item[A6 (legal normalization).]
  A certified gradient error satisfies
  \begin{equation}
    \|\nabla\ph_a(x)-\nabla\phi_a(x)\|\le e(x)<1.
    \label{eq:A6}
  \end{equation}
\end{description}

Define
\begin{equation}
  q(e)=\sqrt{2-2\sqrt{1-e^2}}
  \label{eq:q}
\end{equation}
and the conservative witness-position budget
\begin{equation}
  E_p(x)=\frac{\rho}{\rho-h}\,\delta_a(x)+h q(e(x)).
  \label{eq:Ep}
\end{equation}
\paragraph{Reach-sensitive design margin.}
The factor \(\rho/(\rho-h)>1\) is not required by
\Cref{lem:projection}, which gives the sharper displacement bound
\(\delta_a+hq(e)\).  We retain it as a conservative convention that grows as
the query approaches the reach limit.  Soundness uses only
\(E_p\ge\delta_a+hq(e)\), and all reported certification results use this
convention.
Let \(\mathcal{B}_x\) be a convex box containing the segment
\([\widehat p(x),\pp(x)]\), constructed by inflating the source cell around
\(\widehat p\) by \(E_p\) in every coordinate.  Let \(L_b(x)\) be a certified
Lipschitz constant for \(\ph_b\) on \(\mathcal{B}_x\).  The transferred
visibility margin is
\begin{equation}
  S(x)=\ph_b(\widehat p(x))-\eps_b-L_b(x)E_p(x).
  \label{eq:margin}
\end{equation}
The full field contract also records A3.  The witnessed deficiency certificate
accepts a point only when A1, A2, A4--A6, and \(S(x)>0\) have all been
verified.

\subsection{Tight normal and witness bounds}

\begin{lemma}[Tight normalized-direction bound]
\label{lem:normal}
Let \(\|n\|=1\) and \(\|u-n\|\le e<1\).  Then
\begin{equation}
  \left\|\frac{u}{\|u\|}-n\right\|
  \le q(e)=\sqrt{2-2\sqrt{1-e^2}},
  \label{eq:normal-bound}
\end{equation}
and the bound is attained.
\end{lemma}

\begin{proof}
The vector \(u\) lies in the ball of radius \(e\) centered at \(n\).  The
largest possible angle occurs when the ray through \(u\) is tangent to this
ball, so \(\sin\theta_{\max}=e\).  The two unit directions are then separated
by the chord \(2\sin(\theta_{\max}/2)\), which is exactly
\eqref{eq:normal-bound}.  Equality is reached in this tangent case.
\end{proof}

The commonly obtained triangle-inequality bound \(2e/(1-e)\) is valid but
becomes larger than the trivial diameter two at \(e\ge1/2\).  In contrast,
\eqref{eq:normal-bound} remains informative for every legal \(e<1\).

\begin{lemma}[Witness displacement]
\label{lem:projection}
Under A1, A2, A4, and A6,
\begin{equation}
  \|\widehat p(x)-\pp(x)\|
  \le \delta_a(x)+h q(e(x))\le E_p(x).
  \label{eq:projection-bound}
\end{equation}
\end{lemma}

\begin{proof}
Using \eqref{eq:witnesses},
\begin{equation}
  \widehat p-\pp
  =\bigl(\phi_a-\ph_a\bigr)n_a
   +\ph_a\bigl(n_a-\nh_a\bigr).
\end{equation}
The first term is bounded by \(\delta_a\); the second by
\(h q(e)\) from \Cref{lem:normal}.  Since \(\rho/(\rho-h)>1\), the operational
budget \eqref{eq:Ep} is a conservative enlargement.  A4 ensures that the true
foot point and the normal used above lie in the legal reach tube.
\end{proof}

\begin{lemma}[Inactive-field transfer]
\label{lem:lipschitz-transfer}
If \(\ph_b\) is \(L_b\)-Lipschitz on a convex set containing
\([\widehat p,\pp]\), then
\begin{equation}
  |\ph_b(\widehat p)-\ph_b(\pp)|
  \le L_b\|\widehat p-\pp\|\le L_b E_p.
  \label{eq:lipschitz-transfer}
\end{equation}
\end{lemma}

\begin{proof}
Apply the Lipschitz inequality and then \Cref{lem:projection}.  Enlarging the
evaluation box by \(E_p\) makes sure that it contains the required segment.
The original, uninflated source cell may not contain that segment, so it would
not complete the proof.
\end{proof}

\subsection{Soundness theorem}

\begin{theorem}[Witnessed Boolean contract soundness]
\label{thm:witness-sound}
Let \(\psi=\min(\ph_A,\ph_B)\).  If A1, A2, A4--A6 and \(S(x)>0\) hold at
\(x\), then
the approximate active branch is also the true active branch,
\(\pp(x)\in\partial(A\cup B)\), and
\begin{equation}
  \dist\!\left(x,\partial(A\cup B)\right)-|\psi(x)|
  \le\delta_a(x).
  \label{eq:main-bound}
\end{equation}
\end{theorem}

\begin{proof}
Assume \(\ph_a(x)\le\ph_b(x)\).

\emph{Branch stability.}
By A2 and A5,
\begin{align}
  \phi_b(x)-\phi_a(x)
  &\ge \bigl(\ph_b(x)-\eps_b\bigr)
       -\bigl(\ph_a(x)+\eps_a\bigr)\\
  &>\gamma>0.
  \label{eq:branch-proof}
\end{align}
Thus the approximate and true \(\argmin\) agree.

\emph{Projection and transfer.}
By \Cref{lem:projection},
\(\|\widehat p-\pp\|\le E_p\), so the construction of
\(\mathcal B_x\) places \(\pp\) in that box.  A2's domain-containment clause
makes both the inactive value budget and its Lipschitz certificate valid at
every point used below.  By
\Cref{lem:lipschitz-transfer},
\begin{equation}
  \ph_b(\pp)\ge\ph_b(\widehat p)-L_bE_p.
\end{equation}
Since \(S(x)>0\), the right-hand side is greater than \(\eps_b\).  A2 now gives
\begin{equation}
  \phi_b(\pp)\ge\ph_b(\pp)-\eps_b>0.
  \label{eq:true-visible}
\end{equation}
Therefore the true primitive foot point lies strictly outside the inactive
solid.  Because \(\pp\in\partial A\), it follows that
\(\pp\in\partial(A\cup B)\).

Finally,
\begin{equation}
\begin{split}
  \dist\!\left(x,\partial(A\cup B)\right)
  &\le\|x-\pp\|=|\phi_a(x)|\\
  &\le|\ph_a(x)|+\delta_a(x)=|\psi(x)|+\delta_a(x),
\end{split}
\end{equation}
which is \eqref{eq:main-bound}.
\end{proof}

\begin{remark}[Fail-closed semantics]
The theorem is conditional; it is not a global approximation claim.  If the
gradient bound reaches one, the branch gap is too small, the enlarged box
cannot be bounded, or the margin is nonpositive, the method returns
\(\Unknown\).  Rejecting such a point reduces completeness but does not break
soundness.
\end{remark}

\begin{corollary}[N-ary hard-union certificate]
\label{cor:nary-union}
Let \(m\ge2\), \(\psi_m=\min_{1\le i\le m}\ph_i\), and let \(a\) be the
selected operand at \(x\).  Assume A1, A4, and A6 for the active operand,
and interpret A2 over all \(m\) operands, including their query-point value
budgets and validity domains.
For every \(j\ne a\), assume the generalized branch condition
\begin{equation}
  \ph_j(x)-\ph_a(x)>\eps_a+\eps_j+\gamma,
  \label{eq:nary-branch}
\end{equation}
domain closure of the same inflated box \(\mathcal B_x\), and a valid
Lipschitz constant \(L_j\) there.  Define
\begin{equation}
  S_j(x)=\ph_j(\widehat p(x))-\eps_j-L_jE_p(x).
  \label{eq:nary-margin}
\end{equation}
If \(S_j(x)>0\) for all \(j\ne a\), then the approximate and true hard unions
select \(a\), \(\pp(x)\in\partial(\bigcup_{i=1}^m A_i)\), and
\begin{equation}
 \dist\!\left(x,\partial\!\left(\bigcup_{i=1}^m A_i\right)\right)
 -|\psi_m(x)|\le\delta_a(x).
\end{equation}
\end{corollary}

\begin{proof}
Condition~\eqref{eq:nary-branch} applies the branch-stability argument to every
inactive operand.  Lemmas~\ref{lem:projection} and
\ref{lem:lipschitz-transfer}, followed by \(S_j>0\), give
\(\phi_j(\pp)>0\) for every \(j\ne a\).  Thus \(\pp\in\partial A_a\) is
outside all other solids and lies on the boundary of their union.  The final
distance inequality is identical to the last step of
Theorem~\ref{thm:witness-sound}.
\end{proof}

\begin{corollary}[Exact witnessed region]
For exact operands, if the selected primitive foot point is outside the
inactive solid, then the hard-min magnitude equals the true union distance at
that query point.
\end{corollary}

\begin{proof}
Exactness gives \(E_p=0\) and makes the witnessed foot point a union-boundary
point, hence true distance is at most \(|\psi|\).  The hard minimum of exact
signed distances is conservative for the union, giving the reverse
inequality.
\end{proof}

\section{Certified implementation}
\label{sec:implementation}

\subsection{Cell-wise bounds}

The theorem does not depend on a particular representation.  An implementation
only needs to provide valid values of \(e\) and \(L_b\).  Our verifier uses the
following sign-preserving multiplicative family:
\begin{equation}
  \ph_i(x)=\phi_i(x)\bigl(1+\alpha s_\omega(x)\bigr),
  \quad
  s_\omega(x)=\sin(\omega_1x_1)\cos(\omega_2x_2)\sin(\omega_3x_3).
  \label{eq:perturbation}
\end{equation}
Because \(0<\alpha<1\), the multiplier is positive.  Differentiation gives
\begin{equation}
  \nabla\ph_i-
  \nabla\phi_i
  =\alpha s_\omega\nabla\phi_i
   +\alpha\phi_i\nabla s_\omega.
  \label{eq:gradient-difference}
\end{equation}
For an axis-aligned cell \(C\) of half-width \(r_C\), wrapped trigonometric
intervals provide \(\overline{|s_\omega|}_C\) and
\(\overline{\|\nabla s_\omega\|}_C\).  Since an SDF is one-Lipschitz,
\begin{equation}
 e_C\le\alpha\left(
 \overline{|s_\omega|}_C+
 \bigl(|\phi_i(x_C)|+\sqrt{3}r_C\bigr)
 \overline{\|\nabla s_\omega\|}_C\right).
 \label{eq:cell-gradient}
\end{equation}
Evaluating only the cell center would not be valid when
\(\|\omega\|r_C\gtrsim1\).

On the inflated box \(\mathcal{B}_x\), a Lipschitz product bound gives
\begin{equation}
 L_b\le
 \left(1+\alpha\overline{|s_\omega|}_{\mathcal{B}_x}\right)
 +\alpha\overline{|\phi_b|}_{\mathcal{B}_x}
 \overline{\|\nabla s_\omega\|}_{\mathcal{B}_x}.
 \label{eq:Lb-compute}
\end{equation}
This proof uses the global one-Lipschitz property of signed distance and does
not require differentiability of \(\phi_b\) at every point of the box.

The certifier therefore uses bounds supplied by a trusted upstream method.
For the analytic benchmark below, the known perturbation formula provides
these bounds.  Sampled or neural fields would need another certified source
for \(D_a^{\nabla}\), \(D_b^{\mathrm{Lip}}\), \(e\), and \(L_b\).  Possible
sources include interval automatic differentiation, verified interpolation,
and certified Lipschitz networks.  We do not provide a general method for
producing these bounds.

For \eqref{eq:perturbation}, the local active-field value budget used in
\eqref{eq:Ep} is
\begin{equation}
  \delta_a(x)=\frac{\alpha|\ph_a(x)|}{1-\alpha},
  \label{eq:local-value}
\end{equation}
because \(|\ph_a|\ge(1-\alpha)|\phi_a|\).

\subsection{Algorithm}

\Cref{alg:certifier} lists the steps of the pointwise test.  It uses a strict
comparison whenever the proof needs a strict exterior or branch margin.

\begin{algorithm}[H]
\caption{Pointwise witnessed-union certification}
\label{alg:certifier}
\begin{algorithmic}[1]
\Require \(x,\ph_A,\ph_B,\eps_A,\eps_B,\rho,\gamma\), source cell \(C\), certified domains
\State \(a\gets\argmin_{i\in\{A,B\}}\ph_i(x)\); \(b\gets\{A,B\}\setminus\{a\}\)
\If{\(x\notin D_A^{\rm val}\cap D_B^{\rm val}\) or \(C\not\subset D_a^{\nabla}\)}
  \State \Return \(\Unknown\)
\EndIf
\If{\(\ph_b(x)-\ph_a(x)\le\eps_a+\eps_b+\gamma\)} \Return \(\Unknown\) \EndIf
\State certify \(e_C\) using \eqref{eq:cell-gradient}
\If{\(e_C\ge1\)} \Return \(\Unknown\) \EndIf
\State \(h\gets|\ph_a(x)|\); certify \(\delta_a(x)\)
\If{\(h+\delta_a(x)\ge\rho\)} \Return \(\Unknown\) \EndIf
\State \(\nh_a\gets\nabla\ph_a/\|\nabla\ph_a\|\);
       \(\widehat p\gets x-\ph_a\nh_a\)
\State compute \(E_p\) by \eqref{eq:Ep}; inflate \(C\) around \(\widehat p\) by \(E_p\)
\If{\(\mathcal B_x\not\subset D_b^{\rm val}\cap D_b^{\rm Lip}\)} \Return \(\Unknown\) \EndIf
\State certify \(L_b\) on the inflated box
\If{\(\ph_b(\widehat p)-\eps_b-L_bE_p\le0\)} \Return \(\Unknown\) \EndIf
\State \Return \(\Cert(\Delta^-(x)\le\delta_a(x))\)
\end{algorithmic}
\end{algorithm}

The trigonometric interval routine checks the interval endpoints and every
extremum inside the interval.  It handles negative phases, wrapping, and boxes
wider than one period.  Invalid points are masked before evaluation, so the
routine never receives an infinite box; those points remain uncertified.

For the finite n-ary certificate in \Cref{cor:nary-union}, the checks for the
active operand remain unchanged.  The implementation then applies the branch,
domain, Lipschitz-transfer, and strict visibility checks to each inactive
operand.  It returns \(\Cert\) only if all of these checks pass.

\section{Experiments}
\label{sec:experiments}

The experiments check that the implementation follows the proof; they do not
replace \Cref{thm:witness-sound}.  The verification script stops with an error
if any check fails.  These checks cover the global value budget, finite-
difference derivatives, tightness of the normal bound, interval containment,
branch order, and every link in the proof.

\subsection{Geometry and perturbation}

The first experiment uses two balls of radius five, centered at
\((-3,0,0)\) and \((3,0,0)\), inside \([-9,9]^3\).  Their exact SDFs have
closed-form expressions.  We apply the perturbation
\eqref{eq:perturbation} with
\begin{equation}
  (\omega_1,\omega_2,\omega_3)
  =\omega(1,1.108,0.784),
  \qquad \omega\in\{101,301,1001,3001\}.
\end{equation}
The exact maximum operand magnitude on the box is
\begin{equation}
 \phi_{\max}=\sqrt{12^2+9^2+9^2}-5
 =12.492855685.
\end{equation}
We set \(\alpha=0.05/\phi_{\max}=0.0040022875\), so the global value budget is
exactly \(\eps_A=\eps_B=0.05\).  Other parameters are
\(\rho=5\), \(\gamma=0.3\), a reach slack of \(0.2\), and a gradient-norm
floor of \(0.05\).  The source-cell width is \(0.1/\max_i\omega_i\).

For each frequency, rejection sampling with a separate random seed produces
\num{20000} points that satisfy \(|\psi|\le0.1\) and A5.  We use another
uniform sample of \num{40000} points, with seed 7, to scan the band width.
Coverage uses exact one-sided \SI{95}{\percent} Clopper--Pearson lower bounds.

\subsection{Proof-chain validation}

At every certified point, the audit checks the following four proof links and
the branch choice:
\begin{align*}
 &\text{branch: } \argmin\ph_i=\argmin\phi_i,\\
 &\text{link 1: } \|\widehat p-\pp\|\le E_p,\\
 &\text{link 2: } |\ph_b(\widehat p)-\ph_b(\pp)|
                        \le L_b\|\widehat p-\pp\|,\\
 &\text{link 3: } \ph_b(\pp)\ge\eps_b,\\
 &\text{link 4: } \phi_b(\pp)\ge0.
\end{align*}

\begin{table}[t]
\centering
\caption{Targeted \(|\psi|\le0.1\), A5-separated samples.  ``Viol.'' is the
sum of branch-order and four proof-link violations on certified points.  The
confidence limit concerns coverage, not the analytic soundness theorem.}
\label{tab:proof-chain}
\small
\begin{tabular}{@{}rrrrrr@{}}
\toprule
\(\omega\) & samples & certified & coverage & 95\% lower & viol.\\
\midrule
101  & 20,000 & 20,000 & 100.000\% & 99.9850\% & 0\\
301  & 20,000 & 20,000 & 100.000\% & 99.9850\% & 0\\
1001 & 20,000 & 20,000 & 100.000\% & 99.9850\% & 0\\
3001 & 20,000 & 15,738 &  78.690\% & 78.2081\% & 0\\
\midrule
Total & 80,000 & 75,738 & --- & --- & 0\\
\bottomrule
\end{tabular}
\end{table}

\Cref{tab:proof-chain} shows no measured violations among the \num{75738}
certified points.  At frequency 3001, coverage falls because the method
correctly rejects uncertain points.  Its sampled minimum margin is
\(-10.6406\), whereas the minima for
frequencies 101, 301, and 1001 are respectively \(0.2879\), \(0.2798\), and
\(0.1974\).  A negative margin causes rejection rather than a false
certificate.

Before the proof-chain test, four independent guards pass: the analytic value
budget equals \(0.05\); finite differences agree with the implemented
perturbation gradient to a worst relative error of \(1.90\times10^{-8}\); the
observed-to-analytic ratio for \Cref{lem:normal} is \(0.99993\); and 400 wrapped
interval tests show no containment violation.

\begin{figure}[t]
\centering
\includegraphics[width=\linewidth]{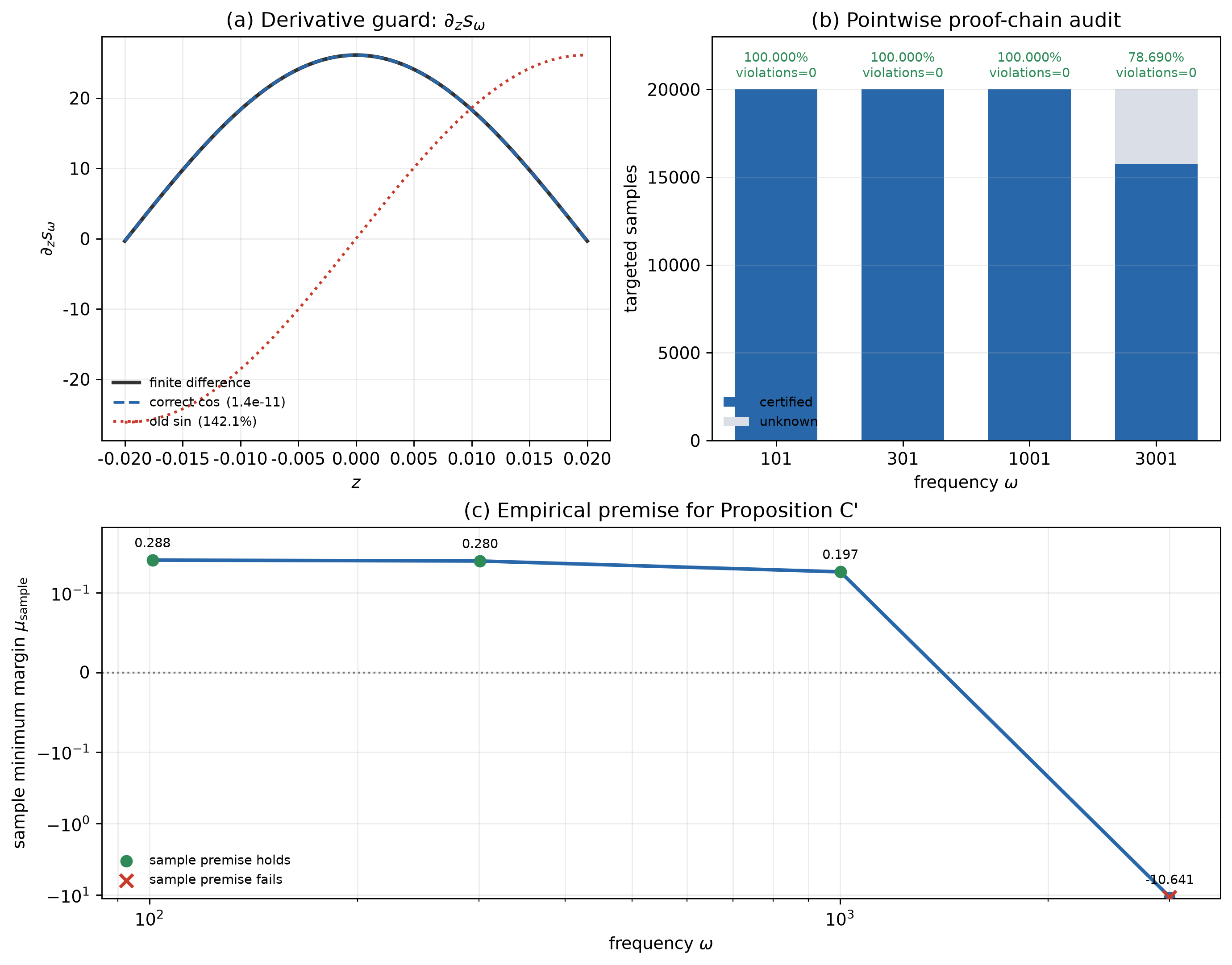}
\caption{Implementation and proof-chain audit.  (a) A finite-difference guard
detects the historical sine/cosine error in the \(z\)-derivative.  (b) Every
certified sample satisfies branch order and all four proof links; high
frequency reduces coverage rather than soundness.  (c) The sampled margin is
positive through \(\omega=1001\) and negative at \(3001\).  Sample minima are
not continuous-domain infima.}
\label{fig:proof-chain}
\end{figure}

\subsection{Band scope and branch separation}

Branch separation is one of the theorem's conditions.  Reporting coverage
only after this filter can make whole-band coverage look too high.  Therefore
\Cref{tab:tau-scan} reports both counts from the same uniform sample.

\begin{table}[t]
\centering
\caption{Band scan at \(\omega=1001\) on \num{40000} uniform samples.  The
A5-separated result is conditional and is never substituted for the whole-band
number.}
\label{tab:tau-scan}
\small
\begin{tabular}{@{}rrrrr@{}}
\toprule
\(\tau\) & whole-band coverage & \(n_{\rm all}\) & A5 coverage & \(n_{\rm A5}\)\\
\midrule
0.05 & 96.16\% & 365   & 100.00\% & 351\\
0.10 & 96.14\% & 725   & 100.00\% & 697\\
0.20 & 93.59\% & 1,404 &  97.55\% & 1,347\\
0.30 & 73.14\% & 2,122 &  76.42\% & 2,031\\
0.50 & 46.74\% & 3,513 &  48.71\% & 3,371\\
\bottomrule
\end{tabular}
\end{table}

At \(\tau=0.1\), 28 of the 725 whole-band samples fail at least one check, so
the entire band is not covered.  All 697 samples that satisfy A5 are
certified.  This shows that branch separation is useful in the sample, but it
does not prove certification of every A5-separated point in the continuous
band.

\subsection{Heterogeneous primitives, Boolean variants, and cost}
\label{sec:expanded-benchmark}

We next test the same contract on exact analytic SDFs for spheres, capsules,
and ring tori.  Five pairings give different curvatures and medial structures.
Each operand is perturbed as in \eqref{eq:perturbation}, with
an operand-specific \(\alpha_i\) chosen from a rigorous boundary-point distance
bound on a contract box one unit wider than the sampling box.  We use
\(\eps_i=0.02\), \(\omega=1001\), \(\gamma=0.08\), and rejection-sample
\num{6000} points per case from \(|\psi|\le0.1\).  We reject a source cell or
enlarged transfer box if it leaves the declared contract domain.  For a
capsule with endpoint \(a\), radius \(r\), and unit axis
\(v\), the bound uses the genuine boundary point
\(q=a+r u_\perp\), where \(u_\perp\perp v\) and \(\|u_\perp\|=1\).

The binary-union rows exercise \Cref{thm:witness-sound}.  For intersection we
replace the minimum and exterior test by a maximum and an interior test; for
difference we use \(\max(\phi_A,-\phi_B)\).  These two sign-dual variants are
checked against exact analytic oracles but are \emph{not} presented as new
theorems.  The final row evaluates a depth-three syntax tree whose four leaves
are combined only by hard union.  Associativity makes it an n-ary minimum, and
the implementation checks the selected witness against every inactive leaf as
required by \Cref{cor:nary-union}.  It does not establish composition for mixed
Boolean trees.

\begin{table}[t]
\centering
\caption{Coverage on heterogeneous analytic narrow bands; \num{6000} samples
per entry.  ``I'' and ``D'' are empirical sign-dual audits.  Every accepted
sample has zero exact-oracle violation.}
\label{tab:expanded-cad}
\small
\begin{tabular}{@{}lrrr@{}}
\toprule
Primitive pairing & Union & Intersection (I) & Difference (D)\\
\midrule
Sphere--sphere   & 98.30\% & 90.65\% & 97.58\%\\
Sphere--capsule  & 98.18\% & 84.75\% & 97.02\%\\
Torus--sphere    & 97.50\% & 88.77\% & 96.92\%\\
Capsule--capsule & 92.80\% & 77.70\% & 89.07\%\\
Torus--capsule   & 92.03\% & 72.67\% & 89.47\%\\
\midrule
Four-leaf nested union & 94.42\% & --- & ---\\
\bottomrule
\end{tabular}
\end{table}

Across all 16 cases, \num{87469} of the \num{96000} band samples are certified,
for a mean coverage of \SI{91.11}{\percent}.  The exact oracle finds no branch,
projection, transfer, or visibility violation among the accepted points.  The
fail-closed checks reject \num{8305} points for insufficient branch separation,
none for the tube/normal condition, and \num{226} for the transferred
visibility margin.  On an Apple M1 Pro, seven warm vectorized NumPy runs for
each \num{6000}-point batch have a median cost of about
\SI{1.5}{\micro\second} per query across the cases; rejection sampling and
file I/O are excluded.  This number measures batch throughput, not the latency
of one scalar query.  The four-leaf union costs
approximately \SI{2.6}{\micro\second} per query.

\begin{figure}[H]
\centering
\includegraphics[width=\linewidth]{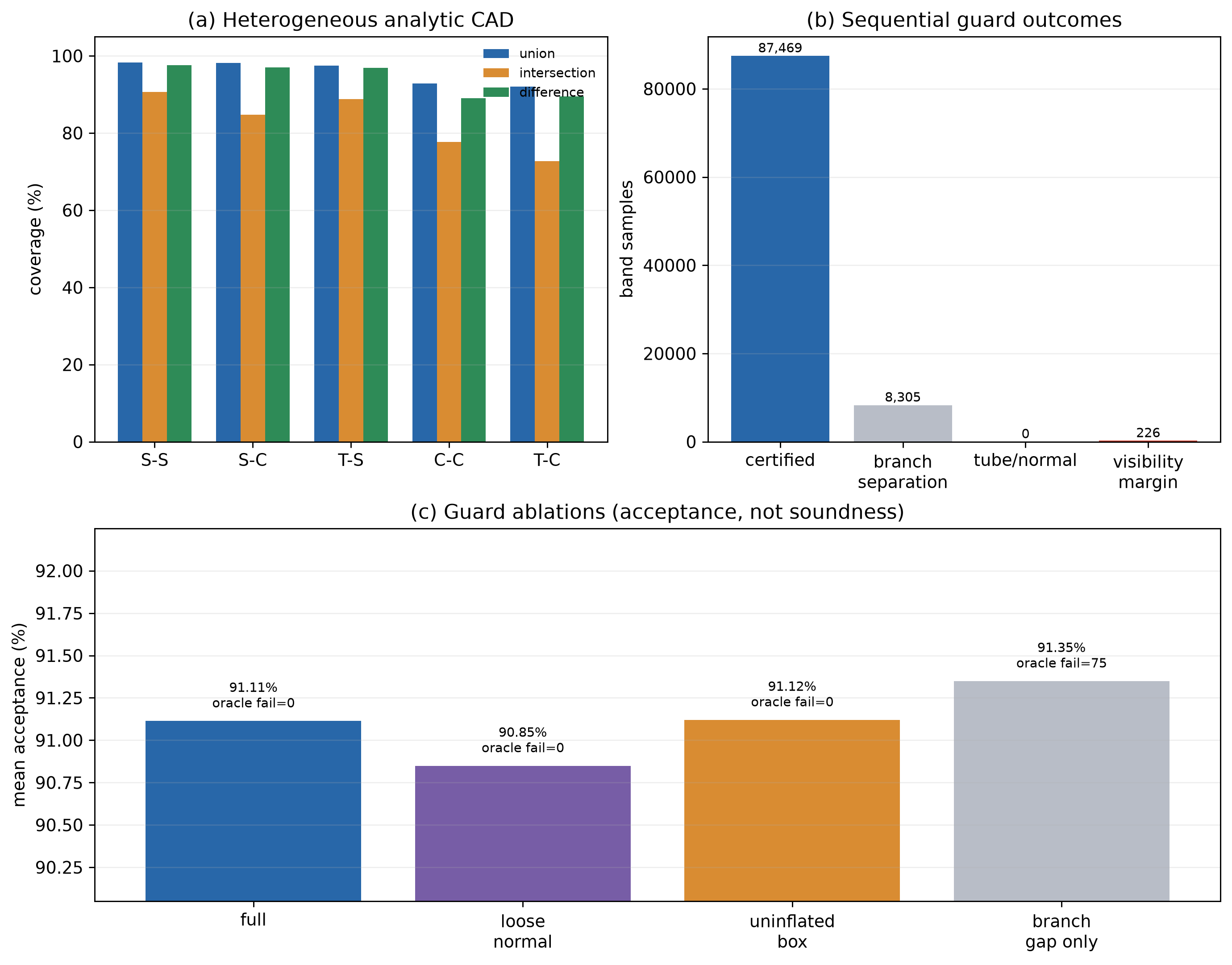}
\caption{Expanded analytic benchmark.  (a) Coverage changes with geometry and
operation.  (b) Most rejected points fail the conservative branch-separation
check.  (c) The tight normal bound raises mean acceptance from 90.85\% to
91.11\%.  Using the original box gives 91.12\% acceptance and no failure in
the sampled oracle, but the box no longer supports the transfer proof.  Using
only the branch gap gives 91.35\% acceptance but admits 75 witnesses that fail
the visibility oracle.  This shows why the witness margin is needed.  All
oracle counts describe only the finite sample.}
\label{fig:expanded-cad}
\end{figure}

\subsection{Conditional narrow-band completeness}

The experiments point to a stronger possibility: there may be a positive
\(\tau^*\) for which every A5-separated point with \(|\psi|\le\tau^*\) is
certified.  We do \emph{not} state this as a theorem.  A proof would require
the following uniform facts on the separated zero level set \(\Gamma\):
\begin{equation}
 \mu_\Gamma=\inf_{\Gamma\cap\mathrm{A5}}S>0,
 \quad \operatorname{Lip}(S)\le K,
 \quad \|\nabla\psi\|\ge m>0,
 \label{eq:Cprime-obligations}
\end{equation}
together with a valid flow or distance argument from the tube to \(\Gamma\).
These would imply \(S(x)\ge\mu_\Gamma-K|\psi(x)|/m\) and hence any
\(\tau<m\mu_\Gamma/K\) would suffice.  The measured
\(\min_iS(x_i)\) is only an order statistic and cannot replace
\(\mu_\Gamma\).  Narrow-band completeness therefore remains a conditional
direction for future work.

\section{Discussion}
\label{sec:discussion}

\paragraph{What the contract recovers.}
A global deficiency bound for a hard Boolean field can be infinite even when
most queries have visible primitive witnesses.  \Cref{thm:witness-sound}
avoids this problem by adding geometric and differential checks and making a
claim only for accepted points.  Its error budget can vary over space and can
be checked by a program.

\paragraph{Why branch gap alone is insufficient.}
A5 guarantees only that the approximate and true \(\argmin\) agree at the
query point.  It does not guarantee that the primitive foot point is visible
on the Boolean boundary.  The transferred margin \(S>0\) checks this second
fact.  Each check therefore protects a different step of the proof.  In the
expanded audit, branch separation alone accepts 75 witnesses that fail the
exact visibility oracle; the full certificate accepts none of them.

\paragraph{Implications for adaptive dual contouring.}
Dual contouring consumes Hermite intersections and normals
\cite{ju2002dual}.  A value-only error budget can certify neither normal
direction nor the validity of the primitive surface that generated a Hermite
sample.  The witnessed contract can serve as a safe filter for such data.
Every accepted sample has a certified point on the Boolean boundary and a
bound on its normal error.  The present paper does not prove that the QEF is
well conditioned or that an adaptive mesh is free of cracks.  Those properties
need additional conditions on cell topology and the linear system.

\paragraph{Proof versus validation.}
The gradient bug found during this study shows why both kinds of checking are
needed.  A wrong derivative once produced no end-to-end violations because the
certificate was conservative.  A finite-difference test and checks of each
proof step exposed the error.  Conversely, even a perfect numerical run cannot
prove a statement for every possible input.  The analytic proof establishes
the general result; the executable audit checks the implementation.

\section{Limitations}
\label{sec:limitations}

First, A1 does not cover sharp primitives or points on the medial axis.  At the
edges and corners of boxes and polyhedra, the reach is zero.  Covering these
cases would require set-valued normals or a projection method that tracks
features.

Second, the paper proves hard union, including the finite n-ary extension in
\Cref{cor:nary-union}.  Intersection and
difference admit analogous witness predicates, but their sign cases and
inactive-boundary conditions should be stated and proved separately.  Deep CSG
trees also require a domain-propagation theorem so that a local certificate is
not silently applied outside its validity set.  The n-ary result therefore does
not establish composition for trees that mix Boolean operators.

Third, \eqref{eq:cell-gradient} and \eqref{eq:Lb-compute} apply specifically to
the multiplicative perturbations used in the benchmark.  Learned or sampled
fields need their own certified derivative and range bounds.  These might come
from interval automatic differentiation, Lipschitz networks, or verified
interpolation.

Fourth, the experiments include spheres, capsules, and ring tori, but all
fields are analytic and use the same type of multiplicative perturbation.  The
results therefore do not establish coverage for sampled grids, neural implicit
fields, sharp CAD parts, or large heterogeneous models.  In those settings,
coverage and cost may change with Boolean depth, curvature, branch separation,
and the quality of the upstream derivative bounds.

Finally, Proposition C' concerns complete certification of a separated band
with positive width.  It remains conditional on the global quantities in
\eqref{eq:Cprime-obligations}.  The positive sample minima for
\(\omega\le1001\) support the proposition, but they are not a lower bound over
the full continuous domain.

\section{Conclusion}
\label{sec:conclusion}

Checks for approximate implicit fields must match the geometric operation that
uses the field.  Even when the sign is exact, a value-error budget does not
stabilize a witness computed from the normalized gradient.  A larger scalar
tolerance is not enough.  The certificate also needs branch separation, a
normal-error bound, positive reach, and a visibility margin that remains valid
over a certified neighborhood.

The witnessed Boolean contract puts these checks in one explicit test.  When
all analytic checks pass and the margin is positive, the selected primitive
provides a real point on the union boundary.  The true-distance deficiency is
then bounded by the active operand's pointwise value-error budget.  The
implementation returns \(\Unknown\) when a check cannot be proved.  Its
reproducibility tests verify each computational step but do not treat sampled
agreement as a proof.  These results provide a basis for spatially varying
contracts in implicit CAD and for later guarantees on normals and QEF
stability in adaptive contouring.

\section*{Reproducibility and data availability}

The accompanying files include the proof-chain verifier, the heterogeneous CAD
benchmark, the JSON results, and the generated figures.  The verification
scripts use NumPy and Matplotlib.  They compute the Clopper--Pearson bound with
the Python standard library, so SciPy is not needed.  Their success message
also states that sampling is not a proof.  All fixed parameters and random
seeds used in the tables are given in
\Cref{sec:experiments}.

\section*{Declaration of generative AI and AI-assisted technologies}

During preparation of this manuscript, the author used OpenAI Codex to help
organize the paper, audit proofs, review code, check reproducibility, and edit
the language.  The author checked the mathematical arguments, references,
source code, experimental results, and generated text, and accepts
responsibility for the published content.



\begin{thebibliography}{99}

\bibitem{hart1996sphere}
J.~C. Hart,
Sphere tracing: a geometric method for the antialiased ray tracing of implicit
surfaces,
\emph{The Visual Computer} 12~(10) (1996) 527--545.
\href{https://doi.org/10.1007/s003710050084}{doi:10.1007/s003710050084}.

\bibitem{frisken2000adaptive}
S.~F. Frisken, R.~N. Perry, A.~P. Rockwood, T.~R. Jones,
Adaptively sampled distance fields: a general representation of shape for
computer graphics,
in: \emph{Proc. SIGGRAPH}, 2000, pp.~249--254.

\bibitem{ju2002dual}
T.~Ju, F.~Losasso, S.~Schaefer, J.~Warren,
Dual contouring of Hermite data,
\emph{ACM Transactions on Graphics} 21~(3) (2002) 339--346.

\bibitem{museth2013vdb}
K.~Museth,
VDB: high-resolution sparse volumes with dynamic topology,
\emph{ACM Transactions on Graphics} 32~(3) (2013) 27:1--27:22.

\bibitem{bloomenthal1997introduction}
J.~Bloomenthal (Ed.),
\emph{Introduction to Implicit Surfaces},
Morgan Kaufmann, 1997.

\bibitem{pasko1995function}
A.~Pasko, V.~Adzhiev, A.~Sourin, V.~Savchenko,
Function representation in geometric modeling: concepts, implementation and
applications,
\emph{The Visual Computer} 11~(8) (1995) 429--446.
\href{https://doi.org/10.1007/BF02464333}{doi:10.1007/BF02464333}.

\bibitem{wyvill1999extending}
B.~Wyvill, A.~Guy, E.~Galin,
Extending the CSG tree: warping, blending and Boolean operations in an implicit
surface modeling system,
\emph{Computer Graphics Forum} 18~(2) (1999) 149--158.

\bibitem{jones2006distance}
M.~W. Jones, J.~A. B{\ae}rentzen, M.~{\v S}r{\'a}mek,
3D distance fields: a survey of techniques and applications,
\emph{IEEE Transactions on Visualization and Computer Graphics} 12~(4)
(2006) 581--599.

\bibitem{ricci1973constructive}
A.~Ricci,
A constructive geometry for computer graphics,
\emph{The Computer Journal} 16~(2) (1973) 157--160.
\href{https://doi.org/10.1093/comjnl/16.2.157}{doi:10.1093/comjnl/16.2.157}.

\bibitem{shapiro2007semianalytic}
V.~Shapiro,
Semi-analytic geometry with R-functions,
\emph{Acta Numerica} 16 (2007) 239--303.
\href{https://doi.org/10.1017/S096249290631001X}{doi:10.1017/S096249290631001X}.

\bibitem{balint2023operations}
C.~B\'alint, G.~Valasek, L.~Gerg\'o,
Operations on signed distance function estimates,
\emph{Computer-Aided Design and Applications} 20~(6) (2023) 1154--1174.
\href{https://doi.org/10.14733/cadaps.2023.1154-1174}{doi:10.14733/cadaps.2023.1154-1174}.

\bibitem{balint2024cutlocus}
Cs.~B\'alint, R.~B\'an, G.~Valasek,
Computing the cut locus, Voronoi diagram, and signed distance function of polygons,
\emph{Computer Aided Geometric Design} 114 (2024) 102388.
\href{https://doi.org/10.1016/j.cagd.2024.102388}{doi:10.1016/j.cagd.2024.102388}.

\bibitem{marschner2023constructive}
Z.~Marschner, S.~Sell\'an, H.-T.~D. Liu, A.~Jacobson,
Constructive solid geometry on neural signed distance fields,
in: \emph{SIGGRAPH Asia 2023 Conference Papers}, 2023, pp.~1--12.
\href{https://doi.org/10.1145/3610548.3618170}{doi:10.1145/3610548.3618170}.

\bibitem{barbier2025pruning}
W.~Barbier, M.~Sanchez, A.~Paris, E.~Michel, T.~Lambert, T.~Boubekeur,
M.~Paulin, T.~Thonat,
Lipschitz Pruning: hierarchical simplification of primitive-based SDFs,
\emph{Computer Graphics Forum} 44~(2) (2025) e70057.
\href{https://doi.org/10.1111/cgf.70057}{doi:10.1111/cgf.70057}.

\bibitem{chen2026greed}
L.~Chen, S.~Mupparaju, C.~Batty, S.~Sell\'an, O.~Stein,
Greed for the spheres: a signed distance interpolation method,
\emph{arXiv:2605.01919}, 2026.

\bibitem{kohlbrenner2026contouring}
M.~Kohlbrenner, M.~Alexa,
Contouring signed distance fields by approximating gradients,
\emph{Computer Graphics Forum} (2026) e70373.
\href{https://doi.org/10.1111/cgf.70373}{doi:10.1111/cgf.70373}.

\bibitem{kalra1989guaranteed}
D.~Kalra, A.~H. Barr,
Guaranteed ray intersections with implicit surfaces,
\emph{Computer Graphics (SIGGRAPH '89)} 23~(3) (1989) 297--306.

\bibitem{galin2020segment}
E.~Galin, E.~Gu\'erin, A.~Paris, A.~Peytavie,
Segment tracing using local Lipschitz bounds,
\emph{Computer Graphics Forum} 39~(2) (2020) 545--554.
\href{https://doi.org/10.1111/cgf.13951}{doi:10.1111/cgf.13951}.

\bibitem{ban2025generalized}
R.~B\'an, G.~Valasek,
Generalized Lipschitz tracing of implicit surfaces,
\emph{Computer Graphics Forum} 44~(1) (2025) e15202.
\href{https://doi.org/10.1111/cgf.15202}{doi:10.1111/cgf.15202}.

\bibitem{moore1966interval}
R.~E. Moore,
\emph{Interval Analysis},
Prentice-Hall, 1966.

\bibitem{snyder1992interval}
J.~M. Snyder,
Interval analysis for computer graphics,
\emph{Computer Graphics (SIGGRAPH '92)} 26~(2) (1992) 121--130.

\bibitem{duff1992interval}
T.~Duff,
Interval arithmetic and recursive subdivision for implicit functions and
constructive solid geometry,
\emph{Computer Graphics (SIGGRAPH '92)} 26~(2) (1992) 131--138.
\href{https://doi.org/10.1145/142920.134027}{doi:10.1145/142920.134027}.

\bibitem{sharp2022spelunking}
N.~Sharp, A.~Jacobson,
Spelunking the deep: guaranteed queries on general neural implicit surfaces via
range analysis,
\emph{ACM Transactions on Graphics} 41~(4) (2022).

\bibitem{yang2023steik}
H.~Yang, Y.~Sun, G.~Sundaramoorthi, A.~Yezzi,
StEik: stabilizing the optimization of neural signed distance functions and
finer shape representation,
in: \emph{Advances in Neural Information Processing Systems}, 2023.

\bibitem{wang2025hotspot}
Z.~Wang, C.~Wang, T.~Yoshino, S.~Tao, Z.~Fu, T.-M.~Li,
HotSpot: signed distance function optimization with an asymptotically
sufficient condition,
in: \emph{Proc. IEEE/CVF Conference on Computer Vision and Pattern Recognition
(CVPR)}, 2025, pp.~1276--1286.

\bibitem{coiffier2024lipschitz}
G.~Coiffier, L.~Bethune,
1-Lipschitz neural distance fields,
\emph{Computer Graphics Forum} 43~(5) (2024) e15128.
\href{https://doi.org/10.1111/cgf.15128}{doi:10.1111/cgf.15128}.

\bibitem{federer1959curvature}
H.~Federer,
Curvature measures,
\emph{Transactions of the American Mathematical Society} 93~(3) (1959)
418--491.

\end{thebibliography}
\end{document}